\documentclass[11pt]{article}
\usepackage{amsmath,amsthm,amssymb,cite,enumerate}
\usepackage[a4paper,left=0.5in,right=1in]{geometry}
\usepackage[usenames]{color}
\usepackage{graphicx}
\usepackage{epsfig}
\usepackage{dcolumn}
\usepackage{bm}
\usepackage{authblk} 

\begin{document}

\def \d {{\rm d}}

\def \bF {\mbox{\boldmath{$F$}}}
\def \bV {\mbox{\boldmath{$V$}}}
\def \bff {\mbox{\boldmath{$f$}}}
\def \bT {\mbox{\boldmath{$T$}}}
\def \bk {\mbox{\boldmath{$k$}}}
\def \bl {\mbox{\boldmath{$\ell$}}}
\def \bn {\mbox{\boldmath{$n$}}}
\def \bbm {\mbox{\boldmath{$m$}}}
\def \tbbm {\mbox{\boldmath{$\bar m$}}}

\def \T {\bigtriangleup}
\newcommand{\msub}[2]{m^{(#1)}_{#2}}
\newcommand{\msup}[2]{m_{(#1)}^{#2}}

\newcommand{\be}{\begin{equation}}
\newcommand{\ee}{\end{equation}}

\newcommand{\beq}{\begin{eqnarray}}
\newcommand{\eeq}{\end{eqnarray}}
\newcommand{\pa}{\partial}
\newcommand{\pp}{{\it pp\,}-}
\newcommand{\ba}{\begin{array}}
\newcommand{\ea}{\end{array}}

\newcommand{\M}[3] {{\stackrel{#1}{M}}_{{#2}{#3}}}
\newcommand{\m}[3] {{\stackrel{\hspace{.3cm}#1}{m}}_{\!{#2}{#3}}\,}

\newcommand{\tr}{\textcolor{red}}
\newcommand{\tb}{\textcolor{blue}}
\newcommand{\tg}{\textcolor{green}}

\newcommand*\bg{\ensuremath{\boldsymbol{g}}}
\newcommand*\bE{\ensuremath{\boldsymbol{E}}}
\newcommand*\bh{\ensuremath{\boldsymbol{h}}}
\newcommand*\bR{\ensuremath{\boldsymbol{R}}}

\def\a{\alpha}
\def\b{\beta}
\def\g{\gamma}
\def\de{\delta}

\def\E{{\cal E}}
\def\B{{\cal B}}
\def\R{{\cal R}}
\def\F{{\cal F}}
\def\L{{\cal L}}

\def\e{e}
\def\bb{b}

\newtheorem{theorem}{Theorem}[section] 
\newtheorem{cor}[theorem]{Corollary} 
\newtheorem{lemma}[theorem]{Lemma} 
\newtheorem{prop}[theorem]{Proposition}
\newtheorem{definition}[theorem]{Definition}
\newtheorem{remark}[theorem]{Remark}  
\newtheorem{proposition}[theorem]{Proposition}

\title{Universal Black Holes}

\author[1]{Sigbj\o rn Hervik\thanks{sigbjorn.hervik@uis.no}}
\author[2]{Marcello Ortaggio\thanks{ortaggio(at)math(dot)cas(dot)cz}}

\affil[1]{Faculty of Science and Technology, University of Stavanger, N-4036 Stavanger, Norway}
\affil[2]{Institute of Mathematics of the Czech Academy of Sciences, \newline \v Zitn\' a 25, 115 67 Prague 1, Czech Republic}

\maketitle

\abstract{We prove that a generalized Schwarzschild-like ansatz can be consistently employed to construct $d$-dimensional static vacuum black hole solutions in any metric theory of gravity for which the Lagrangian is a scalar invariant constructed from the Riemann tensor and its covariant derivatives of arbitrary order. Namely, we show that, apart from containing two arbitrary functions $a(r)$ and $f(r)$ (essentially, the $g_{tt}$ and $g_{rr}$ components), in any such theory the line-element may admit as a base space {\em any} isotropy-irreducible homogeneous space. Technically, this ensures that the field equations generically reduce to two ODEs for $a(r)$ and $f(r)$, and dramatically enlarges the space of black hole solutions and permitted horizon geometries for the considered theories. We then exemplify our results in concrete contexts by constructing solutions in particular theories such as Gauss-Bonnet, quadratic, $F(R)$ and $F$(Lovelock) gravity, and certain conformal gravities.}

\vspace{.2cm}
\noindent





\section{Introduction}

\label{intro}

The prototypical static black hole geometry is described by the celebrated Schwarzschild line-element,
\beq
	\bg=-f(r)\d t^2+\frac{\d r^2}{f(r)} + r^2h_{ij}(x^k)\d x^i\d x^j ,
	\label{schw}
\eeq
where $f=1-\frac{\mu}{r}$ and $\bh=h_{ij}\d x^i\d x^j$ is the metric of a $2$-dimensional round unit sphere. It has been known for various decades that this vacuum solution of general relativity readily extends to Einstein's gravity with a cosmological constant in $d=n+2$ spacetime dimensions, provided one takes $f=1-\frac{\mu}{r^{d-3}}-\lambda r^2$ and $\bh$ is the metric of an $n$-dimensional round unit sphere \cite{Tangherlini63}. Even more generally, one can also replace $\bh$ by {\em any} $n$-dimensional Einstein space with Ricci scalar $\tilde R=n(n-1)K$ and define $f=K-\frac{\mu}{r^{d-3}}-\lambda r^2$ \cite{GibWil87} (see also \cite{Birmingham99}), giving rise to a much richer family of static ``Einstein'' black holes when $d>5$ (at the price of giving up asymptotic flatness or (A)dS-ness). The property of $\bh$ being Einstein is not only sufficient but also necessary, so that the extensions obtained in \cite{GibWil87,Birmingham99} in fact exhaust the space of black hole solution of the form~\eqref{schw} in general relativity.\footnote{It should be pointed out that there exist also static vacuum black holes which are {\em not} of the form~\eqref{schw} and whose horizons are not Einstein, already in five dimensions \cite{CadWoo01} (see \cite{Hervik04} in higher dimensions). Also  five-dimensional static black rings \cite{EmpRea02prd} with a $S^1\times S^2$ horizon cannot be written in the form~\eqref{schw} (as follows from \cite{PraPra05} and the comments on the Weyl type given below) -- these, however, contain a conical singularity. Additionally, static black strings are also excluded by this ansatz, as they typically possess one (or more) privileged spatial direction(s) and a Kaluza-Klein-like asymptotics.\label{foot_rings}} The particular choice of $\bh$ may affect the stability of the solution \cite{GibHar02}.

In addition to Einstein's gravity, gravity theories described by higher-order Lagrangians also have a long history \cite{Weyl18,Eddington_book}, and appear naturally in studies of quantum gravity \cite{DeWittbook} and in the low energy limit of string theory \cite{SchSch74}. Numerous solutions of the form \eqref{schw} have been obtained for various theories in diverse dimensions, mostly {\em assuming} $\bh$ to be a round sphere (as required if spherical symmetry is assumed) or a space of zero or negative constant curvature (references relevant to the present paper will be given in due course in the following sections). Nevertheless, as in Einstein's theory, it would be desirable to characterize the full space of such solutions for more general theories, and one may wonder whether the same ansatz \eqref{schw} can be extended to more general transverse geometries $\bh$ for arbitrary (diffeomorphism invariant, metric) theories of gravity, perhaps by simply modifying $f(r)$ appropriately. However, it was pointed out in \cite{DotGle05} that simply adding a Gauss-Bonnet term to the Einstein-Hilbert Lagrangian places a strong tensorial constraint on the geometry of $\bh$, thus ruling out many known ``exotic'' Einstein black holes. This observation was later extended to cubic \cite{FarDeh14} and arbitrary (generic) Lovelock theories \cite{Ray15} (see also \cite{OhaNoz15}). It is thus clear that, generically, $\bh$ cannot be an arbitrary Einstein space in a gravity theory different from Einstein's. It remains an open question whether a theory-independent characterization of permitted horizon geometries of static black holes can be given (which would be a natural starting point for obtaining full horizon characterizations for specific theories). In this paper we present new results in this direction.

As is well known (and reviewed briefly in section~\ref{sec_Eins} below), in Einstein's gravity the condition $g_{tt}g_{rr}=-1$ in \eqref{schw} follows from the field equations, however this is not necessarily the case in other theories (see, e.g., \cite{Buchdahl53} for an early result). Therefore, apart from the (partial) freedom in the choice of $\bh$, a further extension of \eqref{schw} consists in considering a more general ansatz with two undetermined functions of $r$, i.e., 
\beq
	\bg=e^{a(r)}\left(-f(r)\d t^2+\frac{\d r^2}{f(r)}\right) + r^2h_{ij}(x^k)\d x^i\d x^j .
	\label{metric1}
\eeq

It is the main purpose of the present paper to obtain a sufficient condition on the metric $\bh$ which enables the ansatz \eqref{metric1} to be consistently employed in any theory of gravity of the form 
\be
 S=\int\d^dx\sqrt{-g}{\cal L}(\bR,\nabla\bR,\ldots) ,
\label{action}
\ee
where ${\cal L}$ is a scalar invariant constructed polynomially from the Riemann tensor $\bR$ and its covariant derivatives of arbitrary order. By this we mean that, whatever theory \eqref{action} one chooses, with the ansatz \eqref{metric1} the corresponding field equations will generically reduce to two ODEs  (the precise form of which will depend on the choice of a particular theory) for two unknown metric functions $a(r)$ and $f(r)$, leaving one with some remaining arbitrariness on $\bh$. These spacetimes will in general describe static black holes (examples for various theories are provided in sections~\ref{sec_Eins}--\ref{sec_conf}) and we will name them {\em universal black holes}.\footnote{Comments similar to those in footnote~\ref{foot_rings} also apply to the metric~\eqref{metric1} and also beyond general relativity.} However, the details (including the precise form of $a(r)$ and $f(r)$) and physical properties of the solutions will naturally depend on the specific theory one is interested in. Since for $n=2,3$ an $n$-dimensional Einstein space is necessarily of constant curvature, our analysis we will be of interest for dimension $d\ge6$ (i.e., $n\ge4$). However, our results will apply also in lower dimensions unless stated otherwise.

In section~\ref{sec_universal} we describe properties of the general static ansatz~\eqref{metric1} and prove the main result (already mentioned above) in Proposition~\ref{prop_E}. This is then used to argue how the field equations of {\em any} theory~\eqref{metric1} simplify when evaluated on the considered ansatz. Near-horizon geometries of possible extremal solutions are also briefly discussed. 
 In the subsequent sections \ref{sec_Eins}--\ref{sec_conf} we add further comments for some gravity theories of particular interest (namely, Einstein, Gauss-Bonnet, Lovelock and $F$(Lovelock), quadratic gravity, $F(R)$ and certain conformal gravities). A few explicit solutions are also constructed which exemplify the general results of section~\ref{sec_universal} and in some cases extend certain solutions already previously known in the case of a constant curvature $\bh$. A short summary and some concluding comments are provided in the last section~\ref{sec_discuss}.
In appendix~\ref{app_universal} we define Riemannian {\em universal} spaces and relate those to isotropy-irreducible homogeneous spaces and to the results of \cite{Bleecker79}. In appendix~\ref{app_RT} we briefly review the Robinson-Trautman form of spacetimes~\eqref{metric1}, which is useful to highlight geometric properties thereof and may be convenient for certain computations. As an example, we also work out the explicit form of the field equations of quadratic gravity in arbitrary dimension.

\section{Black Holes with universal horizons}

\label{sec_universal}

\subsection{Geometry of the ansatz}

\label{subsec_geom}

Let us assume the spacetime metric is of the form~\eqref{metric1}. The spacetime is static in regions where $f(r)>0$ and belongs to the Robinson-Trautman class \cite{RobTra62} (extended to arbitrary $d$ in \cite{PodOrt06}), cf. appendix~\ref{app_RT}. Ansatz \eqref{metric1} (or \eqref{ansatz_RT}) describes a warped product with a 2-dimensional Lorentzian factor and is therefore of Weyl type D (or O) and purely electric \cite{PraPraOrt07,HerOrtWyl13}. However, eq.~\eqref{Rvv} shows that the Ricci (and thus the Riemann) tensor is not aligned, unless $a=$const -- in which case one can take $a=0$, upon rescaling $f$ and $t$.\footnote{More in detail, the Weyl type is generically D(bd) since here $\bh$ is Einstein \cite{PraPraOrt07,HerOrtWyl13} (cf. also proposition~8.16 of \cite{OrtPraPra13rev}), in which case there exist precisely two mWANDs $\pa_t\pm f\pa_r$ (see footnote~15 of \cite{OrtPodZof15}) and thus the Riemann type is G if $a_{,r}\neq0$ and D (aligned) if $a_{,r}=0$. The spacetime is conformally flat iff $\bh$ is of constant curvature and the functions $f$ and $a$ satisfy a differential equation which can be obtained from (40,\cite{PodSva15}).} This thus provides an alternative geometric interpretation of the ``$g_{tt}g_{rr}=-1$'' condition discussed in \cite{Jacobson07} (also meaning that, in such a case, $r$ is linear in $v$, i.e., it is an affine parameter along $\bl$ \cite{PodOrt06,Jacobson07}).

In the rest of the paper, it will be assumed that, in \eqref{metric1}, the transverse Riemannian metric $\bh=h_{ij}(x^k)\d x^i\d x^j$ is an $n$-dimensional {\em universal} space (thus being, in particular, Einstein and with constant scalar invariants) -- see appendix~\ref{app_universal} and references therein for a definition and more details.\footnote{To make the reading of the paper more fluent, let us already mention here that universal spaces turn out to be ultimately equivalent to the well-known isotropy-irreducible homogeneous spaces, which have been thoroughly studied (see, e.g., \cite{Bessebook} and references therein).} This property of $\bh$ will be understood from now on. Quantities with a tilde will refer to the transverse space geometry. We normalise the transverse metric so that 
\be
 \tilde{R}_{ij}=(n-1)Kh_{ij} ,
 \label{def_K}
\ee
which obvioulsy implies $\tilde{R}=n(n-1)K$.

\subsection{Reduced field equations and universality}

\label{subsec_reduced}

The field equations derived from \eqref{action} (neglecting boundary terms) are of the form $\bE=0$, where $\bE$ is a symmetric, conserved rank-2 tensor locally constructed out of $\bg$ and its derivatives \cite{Eddington_book} (cf. also \cite{IyeWal94}). However, for the ansatz~\eqref{metric1} they can be drastically simplified thanks to the following

\begin{prop}
 \label{prop_E}
Consider any symmetric 2-tensor, $\bE$,  constructed from tensor products, sums and contractions from the metric $\bg$, the Riemann tensor $\bR$, and its covariant derivatives. Then for any metric of the form (\ref{metric1}) with $\bh$ universal we have: 
\be  
	\bE=F(r)\d t^2+G(r) \d r^2+H(r) h_{ij}(x^k)\d x^i\d x^j. 
	\label{E}
\ee 
\end{prop}

\begin{proof}
First, let us utilise that the metric is invariant under time-reversal: $t\mapsto-t$. This implies that any curvature tensor\footnote{I.e., any tensor constructed polynomially from the Riemann tensor and its covariant derivatives.} is purely electric (as defined in \cite{HerOrtWyl13}). In particular, $\bE$ is purely electric (and, of course, $t$-independent). For a symmetric 2-tensor this implies $E_{ti}=E_{tr}=0$. 
Second, if $ h_{ij}(x^k)\d x^i\d x^j$ is universal then this is an isotropy irreducibe (locally) homogeneous space (cf. appendix~\ref{app_universal}). This means that the isotropy group acts irreducibly on the tangent space of the transverse metric. These symmetries of $\bh$ can be lifted trivially to the total metric (\ref{metric1}), and hence, the tensor $\bE$ needs to be invariant under the these symmetries as well. Using the isotropy group (which acts irreducibly on $T_pM$), we thus get $E_{ij}\propto h_{ij}$ (cf. also \cite{Wolf68}) and $E_{ri}=0$. Finally, since $ h_{ij}(x^k)\d x^i\d x^j$ is a locally homogeneous space, the components $E_{tt}$, $E_{rr}$ can only depend on $r$, and $E_{ij}=H(r) h_{ij}(x^k)\d x^i\d x^j.$
\end{proof}

The tensorial field equation $\bE=0$ thus reduces to three ``scalar'' equations $F(r)=0$, $G(r)=0$ and $H(r)=0$. However, since $\bE$ is identically conserved, it is easy to see that $H(r)=0$ holds automatically once $F(r)=0=G(r)$ are satisfied (see also \eqref{conservation} for a frame reformulation of this statement). We are thus left with just two ODEs for the two metric functions $a(r)$ and $f(r)$. Their precise form will depend on the particular gravity theory under consideration and is of no interest for the general considerations of this paper (several explicit examples can be found in the references given in sections~\ref{sec_Eins}--\ref{sec_conf};  cf. appendix~\ref{app_RT} for quadratic gravity.).

We further observe that \eqref{E} means that, in a frame adapted to the two mWANDs $\pa_t\pm f\pa_r$, $\bE$ possesses only components of b.w. $\pm2$ and 0 (i.e., $E_{++}=E_{--}=f^{-1}(-E^t_{~t}+E^r_{~r})$, $E_{+-}=-e^af(E^t_{~t}+E^r_{~r})$ and $E_{ij}$).

The following Lemma will also be useful: 
\begin{lemma} 
\label{lemmaD}
If $E$ is constructed from only type D tensors then the mixed tensor components obey:\footnote{In the coordinates \eqref{ansatz_RT}, condition~\eqref{EttErr} is equivalent to $E_{vv}=0$.}
\be 
	E^t_{~t}=E^r_{~r}. 
	\label{EttErr}
\ee
\end{lemma}
\begin{proof}
This follows simply from the fact that any type D tensor has a boost isotropy; hence so must the tensor $E$. 
\end{proof}

Condition~\eqref{EttErr} means that $F(r)=-f^2G(r)$ in \eqref{E}, so that in this case one is left with (at most) one non-trivial field equation (ODE) for the two metric functions $a(r)$ and $f(r)$, thus leaving (at least) one of those {\em undetermined}. Something similar (but not quite the same) occurs in Lovelock theories that admit degenerate vacua (see sections~\ref{sec_GB} and \ref{sec_Lov} below and references therein).

\subsection{Extremal limits and near-horizon geometries}

Let us briefly comment on the near-horizon geometries associated with extremal limits of the universal black holes described above (we refer to the review \cite{KunLuc13} and references therein for definitions and general properties of near-horizon geometries). 

Metric~\eqref{metric1} possesses horizons at zeros of $f(r)$. In the coordinates~\eqref{ansatz_RT}, these correspond to zeros of ${\cal H}(v)$ (where the Killing vector field $\pa_u$ becomes null). We now assume that, at least in certain theories~\eqref{action}, there exist solutions with a {\em degenerate} horizon, i.e., for which (without loosing generality, one can always redefine $v\mapsto v+v_0$ so that such horizon lies at $v=0$)
\be
 {\cal H}(v)=v^2{\cal F}(v) ,
\ee  
where ${\cal F}(v)$ is a smooth function (we also assume $r(v)$ to be such). Then, by rescaling $v\mapsto\epsilon v$, $u\mapsto\epsilon^{-1}u$ and taking the limit $\epsilon\to0$ \cite{KunLuc13} one arrives at the near-horizon line-element
\be
 \bg=-2\d u\d v-2{\cal F}_0\d u^2+ r^2_0h_{ij}(x^k)\d x^i\d x^j , 
 \label{extremal}
\ee
where ${\cal F}_0\equiv{\cal F}(0)$ and $r_0\equiv r(0)$. This is a Nariai-like direct product of dS$_2$ (if ${\cal F}_0<0$) or AdS$_2$ (if ${\cal F}_0>0$) with the IHS (universal) space $\bh$ characterizing the original black hole solution \eqref{ansatz_RT}. It possesses a recurrent null vector field $\pa_v$ and thus belongs to the Kundt class (cf, e.g., the review \cite{OrtPraPra13rev} and references therein).

From the above result, it follows that the near-horizon geometry of the considered extremal black holes is essentially theory-independent (up to fixing the two constants ${\cal F}_0$ and $r_0$). A similar ``universality'' of near horizon geometries was discussed for spherically symmetric spacetimes in \cite{Gurses92}.

\section{Einstein gravity}

\label{sec_Eins}

The purpose of the following sections is to illustrate the results of section~\ref{sec_universal} by giving explicit examples of black holes solutions in various gravity theories of the form~\eqref{action}. As a warm-up, let us start with the simplest case of $d$-dimensional general relativity, for which ${\cal L}=\sqrt{-g}\frac{1}{\kappa}(R-2\Lambda)$ and $E_{\mu\nu}=R_{\mu\nu}-\frac{1}{2}Rg_{\mu\nu}+\Lambda g_{\mu\nu}$. From the (reduced) field equations for~\eqref{metric1} one gets the generalized Schwarzschild-(Anti) de Sitter solutions \cite{Tangherlini63,GibWil87,Birmingham99}: 
\be 
 a(r)=0, \qquad f(r)=K-\frac{M}{r^{d-3}}-\lambda r^2 , 
 \label{f_Einst}
\ee
where
\be
 \lambda=\frac{2\Lambda}{(d-1)(d-2)} .
\ee
In particular, the first of \eqref{f_Einst} follows from \eqref{Rvv}. For later purposes, let us note that \eqref{f_Einst} gives $R=d(d-1)\lambda$.

\section{Gauss-Bonnet gravity}

\label{sec_GB}

This theory is of particular interest in the low-energy limit of string theory \cite{Zwiebach85}. It is defined by the Lagrangian density 
\be
 {\cal L}=\sqrt{-g}\left[\frac{1}{\kappa}(R-2\Lambda)+\gamma I_{GB}\right] , \qquad I_{GB}=R_{\mu\nu\rho\sigma}R^{\mu\nu\rho\sigma}-4R_{\mu\nu}R^{\mu\nu}+R^2 ,
 \label{GB}
\ee
where $\gamma$ is a constant parameter, giving 
\be
 E_{\mu\nu}=\frac{1}{\kappa}\left(R_{\mu\nu}-\frac{1}{2}Rg_{\mu\nu}+\Lambda g_{\mu\nu}\right)+2\gamma\left(RR_{\mu\nu}-2R_{\mu\rho\nu\sigma}R^{\rho\sigma}+R_{\mu\rho\sigma\tau}R_\nu^{\ \rho\sigma\tau}-2R_{\mu\rho}R_\nu^{\ \rho}-\frac{1}{4}I_{GB}g_{\mu\nu}\right) .
 \label{E_GB}
\ee

\subsection{Generic theory}

\label{subsec_GB_gen}

The explicit form of \eqref{E_GB} for the ansatz \eqref{metric1} was given in \cite{DotGle05} and there is no need to reproduce it here. As noticed in \cite{DotGle05}, the field equation $E^t_{~t}-E^r_{~r}=0$ (i.e., $E_{vv}=0$, cf. appendix~\ref{app_RT}) shows that {\em generically} one can set
\be
 a(r)=0 ,
 \label{a=0_GB}
\ee
while integrating the remaining field equation gives \cite{DotOliTro09,Bogdanosetal09,Maeda10,DotOliTro10} 
\be
 f(r)=K+\frac{r^2}{2\kappa\hat\gamma}\left[1\pm\sqrt{1+4\kappa\hat\gamma\left(\frac{2\Lambda}{n(n+1)}+\frac{\mu}{r^{n+1}}\right)-\frac{4\kappa^2\hat\gamma^2\tilde I_W^2}{r^4}}\right] ,
 \label{f_GB}
\ee
where $\mu$ is an integration constant and\footnote{Since $\bh$ is Einstein, using~\eqref{def_K} one finds $\tilde I_{GB}=\tilde C_{ijkl}\tilde C^{ijkl}+n(n-1)(n-2)(n-3)K^2$.\label{footn_GB}}
\be
 \hat\gamma=(n-1)(n-2)\gamma , \qquad n(n-1)(n-2)(n-3)\tilde I_W^2=\tilde C_{ijkl}\tilde C^{ijkl} .
 \label{IW}
\ee

Eq.~\eqref{f_GB} clearly illustrates how the Weyl tensor of the geometry $\bh$ affects the solution and its asymptotic behaviour. The branch with the minus sign admits a GR limit to \eqref{f_Einst} by taking $\hat\gamma\to0$. The non-negative constant $\tilde I_W^2$  vanishes iff $\bh$ is conformally flat (so necessarily when $n=3$), 
in which case one recovers the well-known black holes with a constant curvature base space \cite{BouDes85,Wheeler86_GB,Cai02}. See \cite{DotOliTro09,Bogdanosetal09,Maeda10,DotOliTro10} for properties of the spacetimes with $\tilde I_W^2\neq0$.

In the special case of {\em pure} Gauss-Bonnet gravity (i.e., for $\kappa^{-1}=0$), the above solution is replaced by\footnote{We have redefined the cosmological constant $\Lambda=\kappa\hat\Lambda$ such that it survives for $\kappa^{-1}=0$.}
\be
 a(r)=0 , \qquad f(r)=K\pm\sqrt{\frac{1}{\hat\gamma}\left(\frac{2\hat\Lambda}{n(n+1)}r^4+\frac{\mu}{r^{n-3}}\right)-\tilde I_W^2} ,
 \label{f_GBpure}
\ee
This was obtained in \cite{DadPon15} for the special case when $\bh$ is a product of two equal spheres (cf.~\cite{Wheeler86_GB,Cai04,CaiOht06} when $\tilde I_W^2=0$). As follows from an observation in \cite{Buenoetal16}, it is interesting to note that metric~\eqref{f_GBpure} also solves more a general theory ${\cal L}=\sqrt{-g}(-2\hat\Lambda+\gamma I_{GB}+\eta I_{GB}^{d/4})$ (where $\eta$ is a new coupling constant). The results of \cite{Buenoetal16} further imply that it is also a solution of a general class of theories defined by ${\cal L}=\sqrt{-g}F(I_{GB})$, provided $F'\neq0$ and $8\hat\Lambda/\gamma=(d-4)F/F'$ (both these condition must hold on-shell), where $F'=\pa F/\pa I_{GB}$.

\subsection{Special fine-tuned theories ($I_W^2=0$)}

When $\bh$ is conformally flat and the coupling constants are suitably fine-tuned, in addition to \eqref{a=0_GB}, \eqref{f_GB} there exist also the ``geometrically free'' solutions of \cite{Wheeler86} (cf. also \cite{Whitt88,ChaDuf02,DotOliTro09,Bogdanosetal09,DotOliTro10,MaeWilRay11}) 
\be
 e^{-a(r)}f(r)=K+\frac{r^2}{2\kappa\hat\gamma} , \qquad 8\kappa\hat\gamma\Lambda=-n(n+1) , \qquad \tilde I_W^2=0 ,
\ee
for which the metric function $e^{a(r)}f(r)$ remains undetermined (this cannot occur when $\tilde I_W^2\neq0$ \cite{DotOliTro09,Bogdanosetal09,Maeda10,DotOliTro10}). 

For pure Gauss-Bonnet gravity one has instead
\be
 e^{-a(r)}f(r)=K>0 , \qquad \tilde I_W^2=0 ,
 \label{f_GBpure_special}
\ee
with $e^{a(r)}f(r)$ undetermined, as follows easily from \cite{DotGle05}.

\section{Lovelock gravity}

\label{sec_Lov}

In more than six dimensions, a natural extension of Gauss-Bonnet (and Einstein) gravity is given by Lovelock gravity \cite{Lovelock71}, which retains the second order character of the field equations. The Lagrangian density 
\be
 {\cal L}=\sqrt{-g}\sum_{k=0}^{[(d-1)/2]}c_k{\cal L}^{(k)} , \qquad\qquad {\cal L}^{(k)}=\frac{1}{2^k}\delta_{\mu_1\nu_1\ldots \mu_k\nu_k}^{\rho_1\sigma_1\ldots \rho_k\sigma_k}R_{\rho_1\sigma_1}^{\mu_1\nu_1}\ldots R_{\rho_k\sigma_k}^{\mu_k\nu_k} ,
 \label{Lagr}
\ee
gives \cite{Lovelock71}
\be
 E^\mu_\rho=\sum_{k=0}^{[(d-1)/2]}c_k G^{\mu(k)}_{\rho} , \qquad\qquad G^{\mu(k)}_{\rho}=-\dfrac{1}{2^{k+1}}\delta^{\mu\mu_1\nu_1\ldots \mu_k\nu_k}_{\rho\rho_1\sigma_1\ldots \rho_k\sigma_k}R^{\rho_1\sigma_1}_{\mu_1\nu_1}\ldots R^{\rho_k\sigma_k}_{\mu_k\nu_k} ,
 \label{fieldeqns}
\ee
where $\delta^{\mu_1\ldots \mu_p}_{\rho_1\ldots \rho_p}=p!\delta^{\mu_1}_{[\rho_1}\ldots\delta^{\mu_p}_{\rho_p]}$ and $c_k$ are coupling constants. If $c_0$, $c_1$ and $c_2$ are the only non-zero constants, one recovers the Gauss-Bonnet theory~\eqref{GB}.

Similarly as in section~\ref{subsec_GB_gen}, from $E^t_{~t}-E^r_{~r}=0$ one {\em generically} obtains 
\be
 a(r)=0 ,
\ee
while the remaining field equation determines $f(r)$ as the root of an algebraic equation (generalizing \eqref{f_GB}) that depends on an integration constant, the coupling constants $c_k$ and the Euler invariants of the geometry $\bh$ \cite{Ray15} (while the constraints on a generic transverse space were first obtained in \cite{Ray15}, the fact that it can be consistently taken to be IHS was noticed in \cite{OhaNoz15}). When the base space is a round sphere, one recovers the early results of \cite{Wheeler86_GB} (see \cite{Cai04} for the case of zero and negative constant curvature). However, in the latter case there exist particular choices of the $c_k$ that admit solutions with one undetermined metric function \cite{Wheeler86,Whitt88,MaeWilRay11}.

The field equation determining $f(r)$ simplifies considerably in the case of {\em pure} Lovelock gravity, i.e., when a single coefficient $c_{\bar k}$ for $k=\bar k>0$ (plus a possible cosmological term $c_0$) is non-zero in \eqref{Lagr} \cite{Ray15}. Let us just present an example for which this equation can be integrated explicitly. Namely, in the case of the pure cubic theory ($\bar k=3$, which requires $d\ge7$), using Cardano's formula one can solve the general equation given in \cite{Ray15} to obtain\footnote{In order to arrive at~\eqref{f_cubic} we used the identities ${\cal \tilde L}^{(1)}=\tilde R$ and ${\cal \tilde L}^{(2)}=\tilde I_{GB}$, and the fact that for any $n$-dimensional Einstein space $\tilde R=n(n-1)K$  and $\tilde I_{GB}$ is as in footnote~\ref{footn_GB}. Furthermore, it was useful to write 
\[
  {\cal \tilde L}^{(3)}=4\tilde C_{ijkl}\tilde C^{klmn}\tilde C_{mn}^{\phantom{mn}ij}+8\tilde C_{ijkl}\tilde C^{mjkn}\tilde C^{i\phantom{mn}l}_{\phantom{i}mn}+(n-4)(n-5)K\left[3\tilde C_{ijkl}\tilde C^{ijkl}+n(n-1)(n-2)(n-3)K^2\right] ,
\]
which can be obtained easily using (19,\cite{OlivaRay10_prd}).}
\beq
f(r)-K= & & \frac{1}{(2\hat c_3)^{1/3}}\left[c_0r^6-\frac{\mu}{r^{n-5}}+\hat c_3\tilde J_W+\sqrt{\left(c_0r^6-\frac{\mu}{r^{n-5}}+\hat c_3\tilde J_W\right)^2+4\hat c_3^2\tilde I_W^6}\right]^{1/3} \nonumber \\
			 & & {}+\frac{1}{(2\hat c_3)^{1/3}}\left[c_0r^6-\frac{\mu}{r^{n-5}}+\hat c_3\tilde J_W-\sqrt{\left(c_0r^6-\frac{\mu}{r^{n-5}}+\hat c_3\tilde J_W\right)^2+4\hat c_3^2\tilde I_W^6}\right]^{1/3} , \label{f_cubic}
\eeq
where $\mu$ is an integration constant and we have defined $I_W^2$ as in~\eqref{IW} and 
\beq
 & & \hat c_3=(n+1)n(n-1)(n-2)(n-3)(n-4)c_3 , \\ 
 & & (n-1)(n-2)(n-3)(n-4)(n-5)\tilde J_W=4\tilde C_{ijkl}\tilde C^{klmn}\tilde C_{mn}^{\phantom{mn}ij}+8\tilde C_{ijkl}\tilde C^{mjkn}\tilde C^{i\phantom{mn}l}_{\phantom{i}mn} .
\eeq

Solution~\eqref{f_cubic} was obtained in \cite{DadPon15_JHEP} for the special case when $\bh$ is a product of two equal spheres (a solution for cubic Lovelock theory including lower order curvature terms was obtained earlier in \cite{FarDeh14}). When $I_W^6=0$ ($\Rightarrow J_W=0$) the base space is of constant curvature and one recovers the solution obtained in \cite{CaiOht06} (see also \cite{Ray15,Buenoetal16}). From \cite{Buenoetal16} it follows that \eqref{f_cubic} is also a solution of the theory ${\cal L}=\sqrt{-g}(c_0+c_3{\cal L}^{(3)}+\eta{\cal L}^{(3)d/6})$, as well as of the class of theories defined by ${\cal L}=\sqrt{-g}F({\cal L}^{(3)})$, provided (on-shell) $F'\neq0$ and $6c_0/c_3=-(d-6)F/F'$.

\section{Quadratic gravity}

\label{sec_QG}

Apart from the special case of Gauss-Bonnet gravity, actions quadratic in the curvature have been studied for a long time \cite{Weyl18,Eddington_book,Lanczos38,Gregory47,Buchdahl48}. The most general such theory is defined by
\be
  {\cal L}= \sqrt{-g} \left[\frac{1}{\kappa} \left( R - 2 \Lambda \right)+ \alpha R^2 + \beta R_{\mu\nu}R^{\mu\nu}
      + \gamma I_{GB} \right] ,
  \label{QG}
\ee
with $I_{GB}$ as in \eqref{GB} and $\alpha$, $\beta$, $\gamma$ are constant parameters.  This gives rise in general to field equations the fourth order, namely \cite{Buchdahl48} (we follow the notation of \cite{DesTek03})
\beq
  E_{\mu\nu}=\frac{1}{\kappa} \left( R_{\mu\nu} - \frac{1}{2} R g_{\mu\nu} + \Lambda g_{\mu\nu} \right)
  + 2 \alpha R \left( R_{\mu\nu} - \frac{1}{4} R g_{\mu\nu} \right)
  + \left( 2 \alpha + \beta \right)\left( g_{\mu\nu} \Box - \nabla_\mu \nabla_\nu \right) R \nonumber \\
	{}+ \beta\left[\Box\left( R_{\mu\nu} - \frac{1}{2} R g_{\mu\nu} \right)
  + \left(2R_{\mu\rho\nu\sigma} - \frac{1}{2} g_{\mu\nu} R_{\rho\sigma} \right) R^{\rho\sigma}\right]  \nonumber \\
  {}+ 2 \gamma \left( R R_{\mu\nu} - 2 R_{\mu\rho\nu\sigma} R^{\rho\sigma}
  + R_{\mu\rho\sigma\delta} R_{\nu}^{\phantom{\nu}\rho\sigma\delta} - 2 R_{\mu\rho} R_{\nu}^{\phantom{\nu}\rho}
  - \frac{1}{4} g_{\mu\nu}I_{GB} \right) .
  \label{E_QG}
\eeq

Some of these theories are plagued by ghosts \cite{Stelle77}. As an exception to this, the special subcase $\alpha=0=\beta$ has second order field equations and corresponds to Gauss-Bonnet gravity, already discussed in section~\ref{sec_GB}. For arbitrary values of $\alpha$, $\beta$, $\gamma$ and $\Lambda$, the explicit form of \eqref{E_QG} for the ansatz \eqref{metric1} is given in appendix~\ref{app_RT_QG} using the Robinson-Trautman coordinates~\eqref{ansatz_RT}. Let us discuss now a few subcases of special interest.

\subsection{Einstein spacetimes: $d=4$ or $\gamma=0$}

\label{subsec_QG_Einst}

In certain cases, Einstein spacetimes can also solve quadratic gravity. This is always true in four dimensions \cite{Buchdahl48_2,Buchdahl48_3} so that the $d=4$ Schwarzschild-(A)dS black holes \eqref{f_Einst} are solutions of  quadratic gravity (with $\lambda=\Lambda/3$) \cite{Eddington_book,Buchdahl48_2}. For arbitrary $d$, the form of \eqref{E_QG} when $\bg$ is Einstein has been given in (6,\cite{MalPra11prd}). For our ansatz \eqref{metric1} with \eqref{f_Einst}, one easily sees that, if $(d-4)\gamma\neq4$, only spacetimes of constant curvature are possible (since $R_{\mu\nu\rho\sigma}R^{\mu\nu\rho\sigma}$ must be a constant). Therefore, for $d\neq4$ black hole solutions of this form  can only occur for quadratic gravities with $\gamma=0$. In that case, (6,\cite{MalPra11prd}) with \eqref{f_Einst} reduces to a single condition fixing the effective cosmological constant $\lambda$
\be
 2\kappa^{-1}\Lambda=(d-1)\lambda\left[\kappa^{-1}(d-2)+(d-1)(d-4)\lambda(d\alpha+\beta)\right] .
\label{QG_Einst}
\ee
For $(d-4)(d\alpha+\beta)\neq0$, this is a quadratic equation for $\lambda$, therefore there exist two distinct Einstein black holes \eqref{f_Einst}, except in the degenerate case $(d-2)^2\kappa^{-1}+8(d-4)(d\alpha+\beta)\Lambda=0$, for which they coincide. For $(d-4)(d\alpha+\beta)=0$ there exists a single black hole. In all cases, the transverse metric $\bh$ can be any Einstein space (not necessarily IHS). These Einstein black holes in arbitrary dimension were already discussed for $\Lambda=0$ in \cite{MigWil92} and in the case when $\bh$ is of constant curvature in \cite{NojOdi01}.

\subsection{Pure $R^2$ theory}

\label{subsec_R2}

Apart from Gauss-Bonnet gravity, the simplest quadratic gravity theory is obtained by setting $\kappa^{-1}=\beta=\gamma=0$ in \eqref{QG}, i.e., by considering ${\cal L}= \alpha\sqrt{-g}R^2$ (this is clearly also a subset of $F(R)$ gravity \cite{Buchdahl70}, cf. also section~\ref{sec_F_pol}). The field equations reduce to \cite{Gregory47,Buchdahl48}
\beq
  R \left( R_{\mu\nu} - \frac{1}{4} R g_{\mu\nu} \right)+\left( g_{\mu\nu} \Box - \nabla_\mu \nabla_\nu \right) R=0 .
	\label{R2}
\eeq

Clearly, special solutions of this theory are given by spacetimes with $R=0$ (in particular, Ricci-flat spacetimes, as noticed in \cite{Pauli19,Buchdahl48_2}). It is also obvious that proper Einstein spacetimes solve \eqref{R2} iff $d=4$ \cite{Buchdahl48_2,Buchdahl48_3} (cf. also section~\ref{subsec_QG_Einst} and \cite{MalPra11prd}).\footnote{This has to do with the fact that the Lagrangian density $\sqrt{-g}R^2$ is scale invariant iff $d=4$, while for general $d$ the same property is shared by $\sqrt{-g}R^{d/2}$ \cite{Buchdahl48_3}.}

As for static black holes, for simplicity we restrict ourselves to the special case $a(r)=0$, i.e., to the ansatz~\eqref{schw} (this is now an extra assumption, as opposed to the Einstein and generic Gauss-Bonnet and Lovelock cases). The field equations~\eqref{R2} can then be easily integrated and give
\be
 f(r)=K-\frac{\mu_1}{r^{d-3}}-\frac{\mu_2}{r^{d-2}} ,
 \label{fR2}
\ee
where $\mu_1$ and  $\mu_2$ are integration constants. This solution has $R=0$, and is asymptotically flat when $\bh$ is a round $S^n$. For $\mu_2=0$ it reduces to \eqref{f_Einst} with $\Lambda=0$, which is Ricci-flat. This generalizes previous results obtained in the case of spherical symmetry for $d=4$ in \cite{DesTek03cqg} and for any $d$ and any $\bh$ of constant curvature in \cite{Hendi10}.\footnote{The higher-dimensional spherically symmetric solution~(26) of \cite{DesTek03cqg} appears to be incorrect, cf. \cite{Hendi10} and our~\eqref{fR2}.}  In \cite{Hendi10} it was noticed that, for $d=4s$, the $\mu_2$ term mimics the backreaction of a non-linear, conformally invariant Maxwell term $(F_{\mu\nu}F^{\mu\nu})^s$ in Einstein gravity \cite{HasMar07}. We further observe that, when $d=2p$, it also alternatively mimics the backreaction of a linear electromagnetic $p$-form field \cite{BarCalCha12} (see also \cite{OrtPodZof15}). 

As noticed above, for $d=4$ there is additionally the Einstein solution \eqref{f_Einst} with an arbitrary $\lambda$. Furthermore, some spherically symmetric solutions with $a(r)\neq0$ have been given, e.g., in \cite{CapStaTro08,Kehagiasetal15}.

\subsection{$\Lambda$-$R^2$ theory}

\label{subsec_R2_Lambda}

Adding a cosmological constant gives the theory ${\cal L}= \sqrt{-g} \left(-\frac{1}{\kappa}2\Lambda+ \alpha R^2\right)$, for which the field equations are given by \eqref{R2} with an additional term $\frac{1}{2\kappa\alpha}\Lambda g_{\mu\nu}$ on the LHS. The ansatz $a(r)=0$ leads only to the Einstein solution \eqref{f_Einst} with 
\be
 2\Lambda=d(d-1)^2(d-4)\alpha\kappa \lambda^2 .
\ee
This exists only for $d>4$ (more generally, no Einstein spacetimes solve this theory when $d=4$).

\subsection{Einstein-$R^2$ theory}

It is also natural to consider adding the $R^2$ term to Einstein's theory, i.e., ${\cal L}= \sqrt{-g} \left[\frac{1}{\kappa} \left( R - 2 \Lambda \right)+ \alpha R^2\right]$. Clearly here 
\be
  \frac{1}{\kappa} \left( R_{\mu\nu} - \frac{1}{2} R g_{\mu\nu} + \Lambda g_{\mu\nu} \right)
  + 2 \alpha\left[R \left( R_{\mu\nu} - \frac{1}{4} R g_{\mu\nu} \right)+\left( g_{\mu\nu} \Box - \nabla_\mu \nabla_\nu \right) R\right]=0 .
	\label{Eins_R2}
\ee

It is useful to distinguish between an Einstein and a non-Einstein branch.
\begin{enumerate}[(i)]

\item This theory admits {\em Einstein} black holes \eqref{f_Einst} as solutions, which can be obtained from the solutions discussed in section~\ref{subsec_QG_Einst} by setting $\beta=0$ therein (cf. also \cite{MigWil92,NojOdi01}).

\item In search for {\em non-Einstein} solutions, as in sections~\ref{subsec_R2}, \ref{subsec_R2_Lambda} let us make the simplifying assumption $a(r)=0$. The field equations then imply that in \eqref{schw} one has
\be
 f(r)=K-\frac{\mu_1}{r^{d-3}}-\frac{\mu_2}{r^{d-2}}-\lambda r^2 , \qquad \lambda=-\frac{1}{2d(d-1)\kappa\alpha} ,
\label{fEins_R2_special}
\ee
and that the {\em fine-tuning} of the parameters\footnote{More generally, for $d>4$ this fine-tuning is a necessary and sufficient condition for a spacetime with $R=$const to be a solution of \eqref{Eins_R2} (which fixes $R=4\Lambda$). For $d=4$ this condition is only sufficient, since in this particular dimension also all Einstein spacetimes solve \eqref{Eins_R2} identically \cite{Buchdahl48_2,Buchdahl48_3} (without any fine-tuning).}
\be
 8\kappa\alpha\Lambda=-1 ,
\ee
must additionally hold (otherwise $\mu_2=0$ and one is back in the Einstein case). It is worth observing that these non-Einstein solutions occur precisely at the critical point identified in \cite{NojOdi01} (in the case $\mu_2=0$). The analogy with $p$-form solutions mentioned after~\eqref{fR2} holds also here (solution~\eqref{fR2} can be thought as the limit $\kappa^{-1},\Lambda\to0$ of \eqref{fEins_R2_special}). When $\bh$ is constant curvature, these solutions were obtained in  \cite{HenEslMou12}, and they are asymptotically (A)dS when $\bh$ is a round $S^n$.

\end{enumerate}

In both the above cases $R=d(d-1)\lambda$ is a constant.

\section{$F(R)$ gravity}

\label{sec_F_pol}

Going beyond second order in powers of the curvature, a relatively simple and widely explored theory is given by $F(R)$ gravity, which was originally considered from a cosmological viewpoint \cite{Buchdahl70} but has subsequently been considered also in the context of black hole physics (see, e.g., \cite{MigWil92} and further references given below). The Lagrangian density
\be
 {\cal L}= \sqrt{-g} F(R) ,
\ee
in the metric approach gives rise to the equations of motion (in general of 4th order) \cite{Buchdahl70}
\beq
  \left(R_{\mu\nu}+g_{\mu\nu}\Box-\nabla_\mu\nabla_\nu\right)F'-\frac{1}{2}Fg_{\mu\nu}=0 ,
	\label{f(R)}
\eeq
where $F'=\pa F/\pa R$.

From now on we will only consider polynomial theories 
\be
 F(R)=\sum_{k=0}c_kR^k ,
 \label{F_pol}
\ee
where the $c_k$ are constants and the sum extends to an arbitrary natural number (or is infinite if we simply assume $F$ to be analytic at $R=0$). For simplicity, we shall further restrict ourselves to spacetimes with $R=$const. The field equations~\eqref{f(R)} then generically imply that the spacetime is {\em Einstein}, with the constant value of $R$ determined by (cf. \cite{BarOtt83} when $d=4$)
\be
 \sum_{k=0}\left(d-2k\right)c_kR^k=0 .
\ee
($R$ remains arbitrary in the scale invariant case where $c_{d/2}$ is the only non-zero coefficient in \eqref{F_pol} \cite{Buchdahl48_3}, cf. also \cite{delDobMar09}.) With our ansatz~\eqref{metric1}, this gives the Einstein black hole \eqref{f_Einst}, as discussed in \cite{MigWil92} for $c_0=0$ (see \cite{BueCan17} for some  comments) -- the Weyl tensor of $\bh$ does not enter the field equations (cf.~\eqref{R01}--\eqref{R} and  \eqref{DR}--\eqref{BoxR}) so that $\bh$ can be any Einstein space (not necessarily IHS).

However, for theories such that (for a particular choice of $R=$const) $F(R)=0=F'(R)$ (cf. \cite{NojOdi14} for related comments), i.e.,
\be
 \sum_{k=0}c_kR^k=0=\sum_{k=0}kc_kR^{k-1} ,
 \label{tuning_F}
\ee
{\em non-Einstein} spacetimes with $R=$const are also solutions (eq.~\eqref{tuning_F} can be used to fix $R$ and the bare cosmological constan $c_0$). For example, with the ansatz~\eqref{schw} one obtains $f$ as in \eqref{fEins_R2_special}, with $\lambda$ determined by $R=d(d-1)\lambda$ -- solutions of this type when $\bh$ is of constant curvature have been obtained for certain $F(R)$ theories in \cite{Hendi10,HenEslMou12} (see also \cite{NojOdi14} in four dimensions). 
This solution reduces to \eqref{fR2} in the case $R=0$, for which the fine-tuning is simply $c_0=0=c_1$ (i.e., \eqref{fR2} represents static black holes for all theories~\eqref{F_pol} of quadratic or higher order, as discussed in section~\ref{subsec_R2} in a special case).

A non-Einstein solution with $a(r)\neq$const is given, for example, by \eqref{metric1} with
\be
 e^{a(r)}f(r)=1 , \qquad f(r)=K-\frac{\mu}{r^{d-3}}-\lambda r^2 , \qquad R=(d-1)(d-2)\lambda ,
\ee
and with \eqref{tuning_F}, which extends a four-dimensional traversable wormhole of \cite{CalRinSeb18}.

\section{Special conformal gravities}

\label{sec_conf}

In four dimensions, conformal gravity (a subcase of quadratic gravity defined by ${\cal L}= \sqrt{-g} C_{\mu\nu\rho\sigma}C^{\mu\nu\rho\sigma}$) has attracted interest for some time. Apart from the freedom of conformal rescalings, it also possesses the interesting property that all conformally Einstein metrics solve it (in particular, all Einstein metrics) \cite{Buchdahl53}. In six dimensions, the unique polynomial theory with the same property is defined by \cite{LuPanPop13}
\be
{\cal L}= \sqrt{-g}\left(4I_1+I_2-\frac{1}{3}I_3\right)  \qquad (d=6) , 
 \label{conf_6D}
\ee
with
\beq
 & & I_1=C_{\mu\rho\sigma\nu}C^{\mu\alpha\beta\nu}C_{\alpha}{}^{\rho\sigma}{}_{\!\!\beta} , \qquad  I_2=C_{\mu\nu\rho\sigma} C^{\rho\sigma\alpha\beta}
           C_{\alpha\beta}{}^{\mu\nu} , \nonumber \\
 & & I_3=C_{\mu\rho\sigma\lambda}\Big(\delta^\mu_\nu\, \Box +
    4R^\mu{}_\nu - \frac{6}{5} R\, \delta^\mu_\nu\Big) C^{\nu\rho\sigma\lambda}
   + \nabla_\mu J^\mu , 
\eeq
and the divergence term $\nabla_\mu J^\mu$ in $I_3$ can be found in \cite{LuPanPop13}.  
Static black hole solutions of \eqref{conf_6D} with a transverse space of constant curvature were found in \cite{LuPanPop13} and, in particular, a three-parameter subset of solutions corresponds to black holes conformal to the 6D Schwarzschild-(A)dS metric (see the comments below) -- this special family can be thus obtained by simply performing a conformal transformation \cite{LuPanPop13}, with no need to integrate the complicated field equations which follow from~\eqref{conf_6D}.

In arbitrary dimensions, let us assume that an extension of \eqref{conf_6D} is given, i.e., a conformal theory of gravity which admits Einstein spacetimes as solutions (the precise form of this theory is not important for the following observations). By applying a conformal transformation analogous to the one discussed in \cite{LuPanPop13} to the Einstein black hole metric~\eqref{f_Einst}, one can extend the special solutions of \cite{LuPanPop13} to higher dimensions and to arbitrary Einstein transverse spaces. One thus obtains a conformally Einstein solution given by \eqref{schw} with\footnote{In order to simplify the notation, we call again here $r$ the radial coordinate, although it is not the same $r$ which was used in \eqref{f_Einst} (i.e., before the conformal transformation).}
\be
 f(r)=a_0r^2+a_1r+a_2-\mu\sum_{k=0}^{d-4}\binom{d-1}{k}\frac{c^k}{r^{d-3-k}} ,
 \label{f_conf}
\ee
where
\be
 a_0=Kc^2-\mu c^{d-1}-\lambda , \qquad a_1=2Kc-\mu(d-1)c^{d-2} , \qquad a_2=K-\mu c^{ d-3}\frac{(d-1)(d-2)}{2} ,
\ee
and $c$ is an additional parameter introduced by the conformal transformation. In general, all powers of $r$ from $r^2$ to $r^{3-d}$ are present in \eqref{f_conf}, and $R$ is not a constant for these solutions. For $d=4$ and $K=1$ one recovers the solution of \cite{Riegert84} (see also \cite{ManKaz89,Buchdahl53}, and \cite{Klemm98_cqg} for an arbitrary $K$), while for $d=6$ and $\bh$ of constant curvature those of \cite{LuPanPop13}. However, in 6D solutions more general than \eqref{f_conf} exist which are not conformally Einstein \cite{LuPanPop13}, and this is likely to be true also in higher dimensions. In order to study those, however, one needs to consider a specific conformal gravity and integrate explicitly its field equations, which goes beyond the scope of this paper.

\section{Discussion}

\label{sec_discuss}

We have shown that a generalization of a Schwarzschild-like ansatz can be consistently employed to find $d$-dimensional static vacuum black hole solutions in {\em any} metric theory of gravity~\eqref{action}. In a nutshell, this consists in replacing the standard spherical base space metric by an arbitrary isotropy-irreducible homogeneous space. This gives rise to large families of static solutions and dramatically enlarges the space of permitted horizon geometries, well beyond the usual case of horizons of constant curvature. Let us emphasize that we arrived at our conclusions even without the need of specifying the explicit form of the underlying equations of motion. Our results thus apply to general higher-derivative theories, for which constraints on the horizon geometry may generically contain an arbitrary number of covariant derivatives of the Riemann tensor (in contrast to the previously studied case of Lovelock gravity, for which the horizon constraints are purely algebraic -- see sections \ref{sec_GB} and \ref{sec_Lov} for references). 

The present paper thus makes a first step towards a theory-independent characterization of permitted horizon geometries of static black holes. This can be clearly used as a starting point for obtaining full horizon characterizations for specific theories (which may differ from our conclusions in the sense that certain horizon geometries may be permitted in some theories but not in others, thus not being {\em universal}). In this respect, it would be interesting to understand whether the conditions we have obtained are also necessary.

We have also exemplified our results in various theories of gravity which appear to be of considerable interest, but the same methods can be applied in other theories as well. Let us mention, for example, that there is a growing interest also in (higher derivative) modifications of Lovelock's gravity such as quasi-topological gravities \cite{OlivaRay10,OlivaRay10_prd,MyeRob10,Dehghanietal12,Cisternaetal17} (see also, e.g., \cite{BueCanHen20} and references therein), to which our results also apply.

Various results of the present paper can be extended to black hole solutions with matter, such as electromagnetic or scalar fields -- this will be discussed elsewhere. It would also clearly be desirable to understand physical properties of such universal black holes, such as their thermodynamics and stability. These will in general depend on the considered theory. It should be pointed out that, although the metric of the base space does not enter the field equations for the static BH ansatz (as we have shown), it may still affect the stability of the solution (see for example \cite{GibHar02} in Einstein gravity).

\section*{Acknowledgments}

M.O. is grateful to Sourya Ray for useful discussions. S.H. was supported through the Research Council of Norway, Toppforsk
grant no. 250367: \emph{Pseudo-Riemannian Geometry and Polynomial Curvature Invariants:
Classification, Characterisation and Applications.}  M.O. was supported by research plan RVO: 67985840 and research grant GA\v CR 19-09659S.

\renewcommand{\thesection}{\Alph{section}}
\setcounter{section}{0}

\renewcommand{\theequation}{{\thesection}\arabic{equation}}
\setcounter{equation}{0}

\section{Riemannian geometry: universal $\Leftrightarrow$ IHS}

\label{app_universal}

Isotropy-irreducible homogeneous spaces are defined as (quoting, for example, \cite{Bleecker79}):
\begin{definition}[IHS space]
\label{def_IHS}
 An {\em isotropy-irreducible homogeneous space} (IHS) $(M,\bh)$ is a homogeneous space whose isotropy group at a point acts irreducibly on the tangent space of $M$ at that point. 
\end{definition}

Universal spaces were defined in \cite{Coleyetal08}:
\begin{definition}[Universal space]
\label{def_univ}
 A space $(M,\bh)$ is called {\em universal} if any symmetric conserved rank-2 tensor $\bT(h_{ij},\pa_k h_{ij}, \pa_k\pa_l h_{ij},\ldots)$ constructed from sums of terms involving contractions of the metric and powers of arbitrary covariant derivatives of the curvature tensor  (i.e., ``polynomially'') is proportional to $\bh$.
\end{definition}

\begin{remark}
 Both the above definitions are signature-independent, but in the present paper $(M,\bh)$ is assumed to be a Riemannian space (this has some consequences in the following). The factor of proportionality between $\bT$ and $\bh$ in definition~\ref{def_univ} is necessarily a constant since $\bT$ is conserved. 
\end{remark}

We observe that:
\begin{proposition}
\label{prop_IHS_univ}
	A Riemannian space $(M,\bh)$ is IHS if, and only if, it is universal.
\end{proposition}

\begin{proof} 
The fact that $IHS\Rightarrow$ universal follows from the proof of the theorem in section~15 (p.~137) of \cite{Wolf68}. Conversely, a universal space must be ``semi-solo'' (in the terminology of \cite{Bleecker79}) thanks to  the results of section~39 of \cite{thomas} (cf. also \cite{Anderson84} for further comments and \cite{Gilkey73,AtiBotPat73,Epstein75} for related results about ``natural tensors''). Then, theorem~4.4 of \cite{Bleecker79} implies that it is also IHS.\footnote{The fact that universal$\Rightarrow$ locally homogenous can also be proven in a different way. Namely, from the proof of Theorem~3.2 of \cite{HerPraPra14} it follows that a universal space is CSI, which in turn (since the signature is Riemannian) implies \cite{PruTriVan96} local homogeneity.}

\end{proof}

Some examples of IHS are given in \cite{Bessebook} (see also, e.g., \cite{OhaNoz15} for some comments in the context of Lovelock black holes). The simplest ones are direct products of (identical) spaces of constant curvature, others are given by irreducible symmetric spaces. In dimension $n=4$, an IHS must symmetric and therefore locally one of the following: $S^4$, $S^2\times S^2$, $H^4$, $H^2\times H^2$, $\mathbb{C}P^2$, $H_{\mathbb{C}}^2$, or flat space (cf., e.g., \cite{Bessebook} and references therein).

\section{Robinson-Trautman coordinates}

\label{app_RT}

\subsection{Metric and curvature}

\label{subsec_app_metric}

For any choice of a Riemannian metric $\bh$ (even a non-Einstein one), the line-element \eqref{metric1}, belongs to the $d$-dimensional Robinson-Trautman class \cite{PodOrt06} (here $d=n+2$). This can be easily seen by introducing Eddington-Finkelstein coordinates via 
\be 
	\d t=\d u+(e^{a}f)^{-1}\d v , \qquad \d r=e^{-a}\d v ,
\ee	
such that 
\be
 \bg=-2\d u\d v-2{\cal H}(v)\d u^2+ r^2(v)h_{ij}(x^k)\d x^i\d x^j , \qquad 2{\cal H}=e^af .
 \label{ansatz_RT}
\ee
Both coordinate systems \eqref{metric1} and \eqref{ansatz_RT} can be useful for different purposes. In the coordinates \eqref{ansatz_RT}, let us define the coframe\footnote{With a small abuse of notation, for simplicity we use the $n$-dimensional indices $i,j,k,\ldots$ to label both the $x^i$ coordinates of points of the transverse space and the coframe vectors relative to its metric $\bh$.}
\be
 \bm{\omega}^0=\d v+{\cal H}\d u , \qquad \bm{\omega}^1=\d u , \qquad \bm{\omega}^i=r\bm{\tilde\omega}^{\tilde i} ,
 \label{frame_RT}
\ee
where the $\bm{\tilde\omega}^{\tilde i}$ define a coframe of $\bh$. In the dual frame, one finds the following non-zero Riemann tensor components (see \cite{PodSva15} for the corresponding coordinate components)
\beq
  & & R_{0101}={\cal H}'', \qquad R_{0 i1 j}=r^{-1}(r'{\cal H})'\delta_{ij} , \qquad R_{ijkl}=r^{-2}\tilde R_{\tilde i\tilde j\tilde k\tilde l}-4r^{-2}r'^2{\cal H}\delta_{i[k}\delta_{l]j} , \\
	& & R_{0 i0 j}=-r^{-1}r''\delta_{ij} , \qquad R_{1 i1 j}={\cal H}^2R_{0 i0 j} ,
\eeq
where primes denote differentiation w.r.t. $v$. The Ricci tensor then reads 
\beq
  & & R_{01}={\cal H}''+nr^{-1}(r'{\cal H})' , \qquad R_{ij}=r^{-2}\tilde R_{\tilde i\tilde j}-2r^{-1}\delta_{ij}\left[(r'{\cal H})'+(n-1)r'^2r^{-1}{\cal H}\right] , \label{R01} \\
	& & R_{00}=-nr^{-1}r''=-nr^{-1}(e^{-a})_{,v} , \qquad R_{11}={\cal H}^2R_{00} , \label{Rvv}
\eeq
and the Ricci scalar
\be
  R=r^{-2}\tilde R-2{\cal H}''-4nr^{-1}(r'{\cal H})'-2n(n-1)r'^2r^{-2}{\cal H} .
	\label{R}
\ee
When $\bh$ is Einstein one further has $\tilde R_{\tilde i\tilde j}=\frac{\tilde R}{n}\delta_{\tilde i\tilde j}$, with $\tilde R$=const. 
Note also that $R_{00}=R_{vv}$, and $R_{00}=-(e^{a}f)^{-1}(R^t_{~t}-R^r_{~r})$ in the coordinates \eqref{metric1}.

For certain applications it is also useful to display the first and second non-vanishing covariant derivatives of $R$, namely
\be
 R_{;0}=R' , \qquad R_{;1}=-{\cal H}R' ,
 \label{DR}
\ee
and 
\be
 R_{;00}=R'' , \qquad R_{;01}=-({\cal H}R')' , \qquad R_{;11}={\cal H}^2R'' , \qquad R_{;ij}=2r^{-1}r'{\cal H}R'\delta_{ij} ,
 \label{D2R}
\ee
so that 
\be
 \Box R=2r^{-n}(r^{n}{\cal H}R')' .
 \label{BoxR}
\ee

For brevity, we will not display the first and second covariant derivatives of $R_{\mu\nu}$. Let us only observe that they do not contain the Weyl tensor of $\bh$.

\subsection{General field equations}

It follows from section~\ref{subsec_reduced} that, using the frame defined in section~\ref{subsec_app_metric} and assuming $\bh$ to be IHS, the field equations only possess the two {\em independent} components $E_{00}$ and $E_{01}-{\cal H}E_{00}$. The remaining field equations are not independent and can be expressed as $E_{11}={\cal H}^2E_{00}$ and
\be
 r^{-1}r'E_{kk}=({\cal H}E_{00}-E_{01})'+nr^{-1}r'({\cal H}E_{00}-E_{01})+{\cal H}'E_{00} ,
 \label{conservation}
\ee
where the conservation $E^{\mu\nu}_{\phantom{\mu\nu};\nu}=0$ has been used. 
(Note that $E_{ij}=\frac{1}{n}E_{kk}\delta_{ij}$, cf.~\eqref{E}). It may also be useful to observe that trace of the field equations thus reads
\be
 E^{\mu}_{\phantom{\mu}\mu}= -(n+2)E_{01}+n{\cal H}E_{00}+\frac{r}{r'}\left[({\cal H}E_{00}-E_{01})'+{\cal H}'E_{00}\right] .
 \label{trace}
\ee

\subsection{Field equations for quadratic gravity}

\label{app_RT_QG}

As an example, let us obtain the explicit form of the field equations of quadratic gravity~\eqref{E_QG} for the metric \eqref{ansatz_RT} (with $\bh$ being IHS). In the frame \eqref{frame_RT}, the two independent components $E_{00}$ and $-\frac{1}{n}(E_{01}-{\cal H}E_{00})$ (cf. above) give
\beq
& &  -\frac{n}{\kappa}r^{-1}r''-2\tilde R\left[(n-2)(\a+\g)r^{-3}r''+2(\b+3\a)r^{-4}r'^2\right] \nonumber \\
& & {}+2(2\a+\b)\left[{\cal H}''''+2nr^{-1}r'{\cal H}'''\right]+2n{\cal H}''\left\{2\left(2\b+7\a\right)r^{-1}r''+\left[2(n-5)\a+(n-4)\b)\right]r^{-2}r'^2\right\} \nonumber \\
& & {}+2n{\cal H}'\left\{3(\b+4\a)r^{-1}r'''+\left[4\left(3n-7\right)\a+3(n-3)\b\right]r^{-2}r'r''-(n-2)\left(3\b+8\a\right)r^{-3}r'^3\right\} \label{E00} \\
& & {}+2n{\cal H}\Big\{(\b+4\a)r^{-1}r''''+\left[n(\b+4\a)-4(\b+3\a)\right]r^{-2}r'r'''+\left[2n(\b+4\a)-\b-8\a\right]r^{-2}r''^2 \nonumber \\
& & {}+\left[2(n^2-11n+14)\a+(8-5n)\b+2(n-1)(n-2)\gamma\right]r^{-3}r'^2r''+4(n-1)(\b+3\a)r^{-4}r'^4\Big\}=0 , \nonumber 
\eeq
\beq
& &  \frac{1}{n\kappa}\left[n(n-1)r^{-2}r'^2{\cal H}+nr^{-1}r'{\cal H}'+\Lambda-\textstyle{\frac{1}{2}}r^{-2}\tilde R\right]-\textstyle{\frac{1}{2n}}r^{-4}\left[\g\tilde I_{GB}+\textstyle{\frac{1}{n}}\tilde R^2(\b+n\a)\right]  \nonumber \\
& &  {}-2r'''\left\{(\b+4\a)r^{-1}{\cal H}{\cal H}'+2\left[(n+1)\b+4n\a\right]r^{-2}r'{\cal H}^2\right\} \nonumber \\ 
& &  {}-4r''\left\{(\b+4\a)r^{-1}\left[-\textstyle{\frac{1}{2}}{\cal H}{\cal H}''+{\cal H}'^2\right]+\left[4(2n-1)\a+(2n+\textstyle{\frac{1}{2}})\b\right]r^{-2}r'{\cal H}{\cal H}'+(n-2)\left[(n+1)\b+4n\a\right]r^{-3}r'^2{\cal H}^2\right\} \nonumber \\
 & & {}+2r''^2\left[\b+n(\b+4\a)\right]r^{-2}{\cal H}^2-2r'\left\{\left[(\b+4\a){\cal H}{\cal H}'''+(\b+2\a){\cal H}'{\cal H}''\right]r^{-1}-(n-2)(\a+\g)\textstyle{\frac{1}{n}}\tilde Rr^{-3}{\cal H}'\right\} \label{Ecomb} \nonumber \\
& & {}+r'^2\Big\{-2\left[2n(\b+4\a)+\b\right]r^{-2}{\cal H}{\cal H}''+\left[-n(\b+4\a)+4(\b+3\a)\right]r^{-2}{\cal H}'^2  \\
& &  \qquad\qquad {}+2\left[(n-2)(n-3)\gamma+n\a(n-5)-2\b\right]\textstyle{\frac{1}{n}}\tilde Rr^{-4}{\cal H}\Big\} \nonumber \\
& & {}-2r'^3r^{-3}{\cal H}{\cal H}'\left[2(n-1)(n-2)\g+(\b+6\a)n^2-2(\b+9\a)n+2(-\b+2\a)\right] \nonumber \\
& & {}-2(n-1)r'^4r^{-4}{\cal H}^2\left[(n-2)(n-3)\g+\a n^2-n(\b+9\a)-3\b\right] \nonumber \\
& & {}+\textstyle{\frac{1}{n}}(\b+2\a)\left({\cal H}''^2-2{\cal H}'{\cal H}'''\right)=0 , \nonumber
\eeq
where $\tilde{R}=n(n-1)K$.

One can observe, in particular, that the Weyl tensor of the transverse metric $\bh$ enters only via the term $\tilde I_{GB}$ (which can be written as in footnote~\ref{footn_GB}) in \eqref{Ecomb}. For theories with $\g=0$, therefore, for any IHS $\bh$ the field equations take the same form as in the case when $\bh$ is of constant curvature. In general, considerable simplification of the field equations occurs if one assumes $r''=0$, which fixes one of the metric functions (cf. the comments in section~\ref{subsec_geom}).

For certain applications it may be useful to compute explicitly also the trace~\eqref{trace}, which gives 
\beq
 E^{\mu}_{\phantom{\mu}\mu}=\frac{1}{\kappa}\left[(n+2)\Lambda-\frac{n}{2}R\right]+\frac{1}{2}\left[4(n+1)\a+(n+2)\b\right]\Box R-\frac{1}{2}(n-2)r^{-4}\left[\gamma \tilde I_{GB}+(n\a+\b)\textstyle{\frac{1}{n}}\tilde R^2\right] \nonumber \\
 {}-2n(n-2)\Bigg\{[(n+1)\b+4n\a]r^{-2}r''^2{\cal H}^2+\frac{1}{2n}(2\a+\b){\cal H}''^2+\frac{1}{2}[(\b+8\a+4\g)n-4\g+2\b]r^{-2}r'^2{\cal H}'^2  \nonumber \\
	{}+[n(n-1)\a+(n-1)\b+(n-2)(n-3)\g]r^{-4}r'^2{\cal H}\big[(n-1)r'^2{\cal H}-\textstyle{\frac{1}{n}}\tilde R\big]   \nonumber \\
	{}+{\cal H}''\left[(\b+4\a)r^{-1}r'{\cal H}'+(\a+\g)r^{-2}\big[2(n-1)r'^2{\cal H}-\textstyle{\frac{1}{n}}\tilde R \big]\right] \qquad\qquad \\ 
	{}+[2n\a+2(n-2)\g+\b]r^{-3}r'{\cal H}'\big[2(n-1)r'^2{\cal H}-\textstyle{\frac{1}{n}}\tilde R\big] 		 \nonumber \qquad\qquad \\
	{}+r''\bigg[(\b+4\a)r^{-1}{\cal H}{\cal H}''+[(\b+8\a+4\g)n-4\g+2\b]r^{-2}r'{\cal H}{\cal H}' \nonumber \qquad\qquad \\
	 {}+[2n\a+2(n-2)\g+\b]r^{-3}{\cal H}\big[2(n-1)r'^2{\cal H}-\textstyle{\frac{1}{n}}\tilde R\big]\bigg]\Bigg\} \qquad\qquad 	 \nonumber 
 \label{trace_QG}
\eeq 
where $R$ and $\Box R$ are given in \eqref{R} and \eqref{BoxR}. Note, in particular, that \eqref{trace_QG} is of {\em second order} precisely for the class of theories defined by $4(n+1)\a+(n+2)\b=0$, as observed in more generality in \cite{Farhoudi06,NakOda09} (cf. also \cite{GulTek09,OlivaRay10_prd}). Another choice of special interest is given by the theory $(n+2)\a+\b=0=\g$ \cite{DesTek03} (see also \eqref{QG_Einst}).

\providecommand{\href}[2]{#2}\begingroup\raggedright\endgroup

%

\end{document}